\documentclass[10pt,twocolumn,twoside]{IEEEtran}

\usepackage{hyperref,color} 
\usepackage[noadjust]{cite}
\usepackage{bm}
\usepackage{amsfonts, amsmath, amssymb, bbm}
\usepackage[framemethod=tikz]{mdframed}
\usepackage{graphicx}\usepackage{pgfplots,float,makecell}\usetikzlibrary{arrows}
\usepackage{tikz}\usepackage{textcomp}\usepackage{enumerate}
\usetikzlibrary{shapes.geometric}
\usepackage[textwidth=30mm]{todonotes}

\usetikzlibrary{positioning}
\usepackage{subfigure}
\newtheorem{definition}{Definition}
\newtheorem{theorem}{Theorem}
\newtheorem{corollary}{Corollary}
\newtheorem{lemma}{Lemma}

\newtheorem{proposition}{Proposition}
\newtheorem{example}{Example}
\newtheorem{remark}{Remark}

\def\qed{\hfill \vrule height 7pt width 7pt depth 0pt\medskip}
\newcommand{\ds}{\displaystyle}

\newcommand{\ba}{\begin{array}}
\newcommand{\ea}{\end{array}}

\newcommand{\1}{\mathbbm{1}}
\renewcommand{\l}{\left}\renewcommand{\r}{\right}
\newcommand{\be}{\begin{equation}}
\newcommand{\ee}{\end{equation}}

\newcommand{\ups}{\upsilon}

\newcommand{\mc}{\mathcal}

\newcommand{\E}{\mathbb{E}}

\newcommand{\R}{\mathbb{R}}

\newcommand{\se}{\text{ if }}

\DeclareMathOperator{\tr}{tr}

\def\E{\mathbb{E}}
\def\R{\mathbb{R}}

\begin{document}
 \title{Targeting Interventions for Volatility Reduction in Network Systems---Part I: Idiosyncratic Shocks}
\author{Luca Damonte, Leonardo Cianfanelli, Giacomo Como, and Fabio Fagnani\thanks{Some of the results in the paper appeared in preliminary form in \cite{damonte2020systemic,cianfanelli2026cdc}.}
\thanks{
	L.~Cianfanelli, G.~Como and F.~Fagnani are with the
	Department of Mathematical Sciences “G.L. Lagrange,” Politecnico di Torino,
	10129 Torino, Italy (e-mail: {leonardo.cianfanelli; giacomo.como; fabio.fagnani}@polito.it).
	G.~Como is also with the Department of Automatic Control, Lund University,
	22100 Lund, Sweden.
	L.~Damonte is with the Department of Economics and Finance of LUISS University, 00198 Rome, Italy (e-mail: ldamonte@luiss.it). }
\thanks{This study was carried out within the FAIR - Future Artificial Intelligence Research and received funding from the European Union Next-GenerationEU (PIANO NAZIONALE DI RIPRESA E RESILIENZA (PNRR) – MISSIONE 4 COMPONENTE 2, INVESTIMENTO 1.3 – D.D. 1555 11/10/2022, PE00000013). This manuscript reflects only the authors’ views and opinions, neither the European Union nor the European Commission can be considered responsible for them.}
} 

\maketitle
\thispagestyle{empty}

\begin{abstract}We study optimal intervention design problems aimed at minimizing the aggregate volatility  in network systems. 
We consider linear network equilibrium models emerging in a variety of applications such as network games, production networks, and opinion dynamics. In such interconnected systems, it is well known that exogenous shocks hitting a small subset of nodes can propagate through the network and possibly cause detrimental systemic effects.  
We focus on  idiosyncratic shocks and first analyze the dependence of various measures of the aggregate system volatility on the network structure. In particular, we identify a general network centrality measure that determines the amplification factor of the variance of the individual nodes' shocks to the aggregate system volatility.  We also show how, in certain relevant special cases, such centrality measure reduces to notions that have already been proposed in the network systems literature such as  Katz-Bonacich, cycle, and concentration centralities. 

We then consider optimal network intervention design problems whereby a system planner has the ability to attenuate the effect of the exogenous shocks in some nodes, within a given budget constraint, and aims at minimizing the resulting aggregate system volatility.  We analytically solve the ensuing optimal volatility reduction problems, showing, in particular, that their solution has a \emph{water-filling-type} structure. On the one hand, for low budget levels, the optimal intervention is concentrated on the node(s) $i$ where a quantity $\omega_i$, accounting for the variance of the shock on $i$, the network centrality of $i$, and the cost of the unitary intervention on $i$, is maximized. On the other hand, as the budget increases, the optimal intervention targets more and more nodes, and eventually, for sufficiently large budget levels, it acts on every node $i$ with $\omega_i>0$, with an intensity proportional to $\omega_i$.  

Finally, we study robust versions of the optimal network intervention design problem, modeled as min-max optimization problems  whereby the system planner is interested in designing an intervention that minimizes the aggregate system volatility against a whole set of idiosyncratic shocks. 
These may be interpreted as adversarial problems whereby, given the system planner's choice of the intervention, a malicious adversary selects the worst-case vector of idiosyncratic shocks' variances within a given budget constraint. We provide analytical solutions to such robust optimal network intervention design problems, showing that they also have a water-filling-type structure with the ratio $\bar\omega_i$ between the node's centrality and the square root of the unitary adversarial cost now playing the role previously played by $\omega_i$. 
\end{abstract}
\begin{IEEEkeywords} System volatility, shock propagation, systemic risk, network interventions, volatility reduction.
\end{IEEEkeywords}

\section{Introduction}\label{sec:intro}
The key role of interconnections in propagating and amplifying perturbations in network systems has been widely recognized over the past two decades. It is now well understood that exogenous shocks, even if relatively small and initially localized on a  few nodes, can propagate through the network and possibly cause detrimental systemic effects.
This issue is prevalent across applications ranging from production networks \cite{acemoglu2012network}, financial systems \cite{Glasserman.Young:2015, Massai.ea:22}, social networks \cite{acemouglu2013opinion,Cisneros-Velarde:20}, sensor networks and formation control \cite{bamieh2012coherence,siami2016fundamental}, as well as infrastructure systems \cite{Como-ea:13a,Como-ea:13b}.  

A large body of literature has investigated the role of the network structure in propagating and amplifying the effects of perturbations in large-scale systems. A key finding is that, even if the exogenous shocks hitting the network are idiosyncratic, interconnections introduce significant correlations that may trigger cascading mechanisms and create systemic risk. Of crucial interest is to identify which network structures are the most vulnerable or the most resilient and, within a given network, which nodes are the most critical for the spread of such perturbations. The notion of network centrality, measuring the importance of various nodes in determining the macroscopic performance of the network system, has been identified as crucial in understanding the combined effects of shocks and network interactions \cite{jackson2010social,bloch2016centrality,siami2017centrality, Parise.Ozdaglar:23}.
Such an analysis is instrumental to the effective design of network interventions \cite{ballester2006who,candogan2012optimal,bramoulle2014strategic,zenou2015key,galeotti2020targeting,parise2021analysis,Como.ea:22,bloch2020targeting,cianfanelli2023optimal,Sezer.ea:24,Messina.ea:24}.

In this work, we focus on linear network equilibrium models, with applications including network games \cite{Galeotti.ea:10,jackson2015games,bramoulle2015games,Arditti.ea:24-sep}, production networks \cite{acemoglu2012network,carvalho2014micro,baqaee2019macroeconomic,Carvalho.Tahbaz-Salehi:19}, and opinion dynamics \cite{French:56,degroot1974reaching,friedkin1990social,Como.Fagnani:2016}. We first introduce a class of measures of aggregate system volatility ---parametrized by a symmetric positive semi-definite matrix $Q$ and referred to as $Q$-volatility--- and analyze their dependence on both the network structure and the exogenous shocks' distribution. In particular, in this paper we focus on idiosyncratic exogenous shocks, whereas our companion paper \cite{PartII} addresses the general case where the exogenous shocks hitting the network can be correlated. Our first main contribution consists in identifying a general network centrality measure ---referred to as the $Q$-target centrality--- that determines the amplification factor of the variance of the shocks affecting individual nodes to the system's $Q$-volatility (see Proposition \ref{prop:volatility-uncorrelated}).  We also show how, in certain relevant special cases, the $Q$-target centrality reduces to notions that have already been proposed in the network systems literature such as the Katz-Bonacich \cite{katz1953new, bonacich1987power}, cycle \cite{talamas2017network}, and concentration \cite{Acemoglu.ea:16} centralities. 

We then study optimal network intervention design problems whereby a system planner allocates an available budget to attenuate the effect of the exogenous shocks in some nodes and aims at minimizing the resulting $Q$-volatility of the system. Our second main contribution consists in an analytical solution of the optimal volatility reduction problem (see Theorem \ref{theo:min-uncor}). In particular, we show that this problem has a water-filling type solution: for low budget levels, the optimal intervention is concentrated on the node(s) $i$ where the square root of the product of the $Q$-target centrality times the shock's standard deviation normalized by the square root of the unitary intervention cost (denoted $\omega_i$) is maximized, whereas, as the budget increases, the optimal intervention targets more and more nodes, and eventually, for sufficiently large budget levels, it targets every node $i$ with $\omega_i>0$, with an intensity proportional  to $\omega_i$.  

Finally, we study robust versions of the optimal network intervention design problem, modeled as min-max optimization problems  whereby the system planner is interested in designing an intervention that minimizes the aggregate system volatility against a whole set of idiosyncratic shocks. 
These may be interpreted as adversarial problems whereby, given the system planner's choice of the intervention, a malicious adversary selects the worst-case vector of variances of the idiosyncratic shocks within a given budget constraint. Our third main contribution consists in providing analytical solutions to such robust optimal network intervention design problems (see Theorem \ref{theo:minmax-uncor}), showing that they also have a water-filling-type structure with the ratio $\bar\omega_i$ between the node's centrality and the square root of the unitary adversarial cost now playing the role previously played by $\omega_i$. 

The rest of the paper is organized as follows. This section concludes with a short summary of the main notational conventions to be adopted throughout the paper. Section \ref{sec:model} first presents the considered linear network equilibrium model and the notion of $Q$-volatility, then it discusses how some relevant applications can be framed into our setting. Section \ref{sec:analysis}
analyzes the dependence of the $Q$-volatility of the network equilibrium on the network structure and on the variances of the exogenous idiosyncratic shocks. Section \ref{sec:given-distribution} studies the optimal volatility reduction problem in case of exogenous idiosyncratic shocks with known variances. Section \ref{sec:min-max} studies a robust version of the optimal volatility reduction problem where the system planner is interested in designing a shock attenuation vector that minimizes the system volatility against a set of idiosyncratic shocks. Finally, Section \ref{sec:conclusion} presents some concluding remarks. 

\textbf{Notation} The sets of real, non-negative real, and positive real numbers are denoted by $\R$, $\R_+$, and  $\R_{++}$, respectively. 
For two vectors $a,b$ in $\R^n$, we denote by $a\vee b$ and $a\wedge b$ in $\R^n$ the vectors with entries $(a\vee b)_i=\max \{a_i, b_i\}$ and $(a\wedge b)_i=\min\{a_i, b_i\}$ for $i=1,2,\dots,n$, respectively. The vector of all ones is denoted $\1$, whose size may be deduced from the context. Inequalities between vectors are meant to hold true element-wise. The vector $a^k$ denotes the vector with entries $(a^k)_i = (a_i)^k$. The transpose of a matrix $A$ is denoted by $A'$. 
For a vector $a$ in $\R^n$, we denote by $[a]$ in $\R^{n\times n}$ the diagonal matrix with diagonal entries $[a]_{ii}=a_i$ for $i=1,2,\ldots,n$.
The spectral radius of a square matrix $A$ is denoted by $\rho(A)$. A square matrix $A$ is referred to as Schur stable if $\rho(A)<1$.

\section{Model description and applications}\label{sec:model}

\subsection{Network equilibrium and volatility measures}\label{sec:volatility}
We consider a linear network equilibrium model 
\be\label{equilibrium} x=Ax+u+w\,,\ee
where $x$ in $\R^n$ is a vector, $A$ in $\R^{n\times n}$ is a Schur stable matrix describing the interaction among $n$ nodes, $u$ is a nominal $n$-dimensional input vector, and $w$ is an $n$-dimensional random vector, whose entries $w_i$ will be referred to as the exogenous \emph{shocks}.
The Schur stability of $A$ guarantees that the matrix $I-A$ is invertible and that the equilibrium $x$ is uniquely determined by \eqref{equilibrium} as a function of $u$ and $w$. 
We refer to $A$ as the network \emph{interaction matrix}, 
to 
\be\label{input-output} M=(I-A)^{-1}\,,\ee
as the network \emph{input-output matrix}, and to
\be\label{equilibrium-bis} x=(I-A)^{-1}(u+w)=M(u+w)\,,\ee
as the network \emph{equilibrium}.

Throughout, we shall assume that the exogenous shocks are zero-mean and have bounded second moment, i.e., 
\be\label{noise-assumptions}\E[w_i]=0\,,\qquad \E[w_i^2]<+\infty\,,\qquad i=1,2,\ldots,n\,,\ee
and denote by
\be\label{covariance-matrix}K=\E[ww']\,,\ee
the corresponding \emph{shock covariance matrix}.
In this paper, which constitutes the first part of our work, we shall focus on the special case of \emph{idiosyncratic} shocks, i.e., on the case when $w_i$, $i=1,\ldots,n$, are independent random variables with variance $\sigma_i^2 \ge 0$, so that the shock covariance matrix is positive semi-definite and diagonal:\footnote{In fact, all our results rely on the weaker assumption that the shocks are uncorrelated, so that \eqref{noise-idiosync} holds true.}
\be\label{noise-idiosync}K=[\sigma^2]\,.\ee
In the second part of our work \cite{PartII}, we shall instead study the general case of an arbitrary covariance matrix $K$.

Our aim is first to study the effect of the network interaction matrix $A$ and of the shocks' variances on the network equilibrium $x$ through \emph{aggregate volatility} indices, and then to design optimal intervention policies that (robustly) minimize such network volatility. 
We start by introducing a general network volatility index as follows. For any symmetric positive semi-definite matrix $Q$ in $\R^{n\times n}$, we define 
\be\label{system-volatility}\Phi_Q=\E\left[(x-\E[x])'Q(x-\E[x])\right]\,,\ee
as the \emph{$Q$-volatility} of the network equilibrium \eqref{equilibrium}. 

Three key special cases  are: 
\begin{enumerate}
\item[(i)] $Q=[\delta]$, for some positive vector $\delta$ in $\R_{++}^n$, in which case the $Q$-volatility
\be\label{total-variance}\Phi_{[\delta]}=\sum_{i=1}^{n}\delta_i\E[(x_i-\E[x_i])^2]\,,\ee
coincides with the aggregate weighted variance of the equilibrium vector $x$;  
\item[(ii)] $Q=\gamma\gamma'$ for some non-negative vector 
$\gamma$ in $\R^n_+$, in which case the $Q$-volatility
\be\label{variance-average}\Phi_{\gamma\gamma'}=\E[(\gamma'x-\E[\gamma' x])^2]\,,\ee
coincides with the variance of the weighted aggregate $\gamma'x=\sum_i\gamma_ix_i$ of the equilibrium entries;  
\item[(iii)] $Q=[\delta]-\gamma\gamma'$, for some positive vector $\delta$ in $\R_{++}^n$ and some non-negative vector $\gamma$ in $\R^n_+$ such that $$\sum_{i=1}^n\frac{\gamma_i^2}{\delta_i}\le1\,;$$ 
in particular, when $\delta=\1$ and $\gamma=n^{-1/2}\1$,  
\be\label{polarization}\Phi_{I-\frac1n\1\1'}=\sum_{i=1}^{n}\E[\big((x_i-\bar x) - \E[x_i - \bar x]\big)^2]\,,\ee
coincides with the aggregate variance of the displacement from the average equilibrium value $$\bar x=\frac1n\1'x\,.$$
\end{enumerate}

In the following subsection we show how the three network volatility measures presented above naturally arise in the study of certain network systems. 

\subsection{Applications}
We present some fundamental applications of our framework.

\subsubsection{Linear network dynamics}
Consider $n$ individuals that update their scalar state $x_i(t)$ synchronously according to the linear rule
\be\label{eq:dynamics}
x_i(t+1) = \sum_{j = 1}^n A_{ij} x_j(t) + u_i + w_i, \quad i = 1,2, \cdots,n\,,
\ee
where $A$ in $\R^{n \times n}$ is a Schur stable matrix, $u$ is a reference input vector in $\R^n$ and $w$ is an $n$-dimensional zero-mean random vector with bounded second moment. It is known from \cite{Como.Fagnani:2016} that, for every initial state profile $x(0)$, the state vector $x(t)$ converges to
\be\label{FJ-limit}\lim_{t\to+\infty} x(t)= x=(I-A)^{-1}(u+w)\,.\ee
Hence, the asymptotic state vector of linear network dynamics \eqref{eq:dynamics} with Schur stable matrix $A$ can be framed into our equilibrium model \eqref{equilibrium-bis}.

Basic opinion dynamics models (e.g., the popular Friedkin-Johnsen model \cite{friedkin1990social}) are special cases of the dynamical system described in \eqref{eq:dynamics}, with the vectors $u$ and $w$ interpretable as exogenous signals representing the opinion of stubborn agents present in the modeled community or primitive beliefs playing the role of anchors for the agents' opinions.  In this context, $\Phi_I$ and $\Phi_{\frac1n\1\1'}$ represent, respectively, the aggregate variance and $n$ times the variance of the empirical mean of the asymptotic opinion vector $x$, with respect to the randomness injected by the variable $w$. In this way, they can be interpreted as resiliency measures of the opinion dynamics model against external shocks. Finally, in those scenarios where the asymptotic outcome of the unperturbed model is a consensus configuration (for instance when $Mu = c\1$ for $c$ in $\R$, which in the Friedkin-Johnsen model occurs if and only if all agents $i$ with $\sum_j A_{ij}<1$ have the same anchor opinion), the volatility $\Phi_{I-\frac1n\1\1'}$ measures the distance from the consensus. This can be seen as a measure of the polarization that the shock is producing as the volatility in this case can be rewritten as
$$\Phi_{I-\frac1n\1\1'}=\frac1{2n}\sum_{i=1}^{n}\sum_{j=1}^{n}\E[(x_i-x_j)^2]\,.$$ 

\subsubsection{Production networks}
Consider a simplified model of an economy \cite{berman1994nonnegative} consisting of $n$ interconnected productive sectors, each producing a distinct product. To produce a quantity $x_i$ of good $i$, sector $i$ needs fixed quantities $W_{ij} x_i$ of good $j$, where $W_{ij} \ge 0$ is a fixed constant. In the \emph{open Leontief model}, each product $i$ faces an external demand $d_i$. The market-clearing conditions ensuring that production satisfies exactly both intermediate and external demand are
\be\label{eq:leontief}
x_i = \sum_{j = 1}^n W_{ji} x_j + d_i\,, \quad i = 1,\cdots,n\,.
\ee
Note that, if $W$ is Schur stable and the demands $d$ are random variables with bounded second moment, \eqref{eq:leontief} is a special case of \eqref{equilibrium-bis} with $A = W'$, $u = \E[d]$, and $w = d - \E[d]$. A relevant quantity in this context is the $I$-volatility $\Phi_I$, which measures the aggregate variance of production when the demand is subject to random shocks.

The equilibrium model \eqref{equilibrium-bis} also shows up in other network economics models, such as the general equilibrium theory model based on production functions of Cobb-Douglas type \cite{acemoglu2012network,Carvalho.Tahbaz-Salehi:19}. Denoting the production vector as $y$ and its log transformation as $x=\log y$, the unique \emph{Walrasian} equilibrium is characterized by 
$$x=\log y= (I-W')^{-1}\,w+\kappa\1\,,$$ 
where $W_{ij}$ here represents the \emph{importance} of the output of firm $i$ in the production technology of industry $j$,  $w$ is a vector where $w_i$ represents the productivity log-shock on firm $i$, and where $\kappa$ is a scalar. 

A crucial quantity in this context is the logarithm of the gross domestic product (GDP) of the economy
$$\log (GDP)=\sum_{i=1}^n\gamma_ix_i-\kappa+\sum_{i=1}^n\gamma_i\log\gamma_i\,,$$
where $\gamma_i$ measures the importance of the good produced by industry $i$ for the final consumer.
The $\log (GDP)$ variance coincides with the $\gamma\gamma'$-volatility:  
$$\mathrm{Var}[\log (GDP)]=\E\left[\left(\sum_{i=1}^n\gamma_i\left(x_i-\E[x_i]\right)\right)^2\right]=\Phi_{\gamma\gamma'}\,.$$

\section{Analysis of the $Q$-volatility}\label{sec:analysis}
In this section, we study the dependence of the $Q$-volatility \eqref{system-volatility} of the network equilibrium on the network structure in the special case of idiosyncratic shocks.

We start by introducing the following notion of $Q$-target centrality for the nodes in the network.
\begin{definition} Let $A$ in $\R^{n\times n}$ be a Schur stable interaction matrix. For every symmetric positive semi-definite matrix $Q$ in $\R^{n\times n}$, the $Q$-target centrality of a node $i=1,2,\ldots,n$ is
\be\label{Q-target-centrality}\pi_i=\sqrt{(M'QM)_{ii}}=\left(\sum\nolimits_{j,k}Q_{jk}M_{ji}M_{ki}\right)^{1/2}\,,\ee
where $M$ is the input-output matrix as defined in \eqref{input-output}. 
\end{definition}\medskip

For the three special choices of symmetric positive semi-definite matrix $Q$ introduced in Section \ref{sec:volatility}, the notion of $Q$-target centrality can be related to different concepts previously introduced in the network systems literature. 
Specifically, the following interpretations hold:
\begin{enumerate}
\item[(i)] when $Q=I$, so that the $Q$-volatility coincides with the total variance of the network equilibrium \eqref{total-variance}, we have that the squared $Q$-target centrality  of a node $i$
\be\label{cycle-centrality}\pi_i^2=(M'M)_{ii}=\sum_{j=1}^{n}(M_{ji})^2\,,\ee
is proportional to its so-called cycle centrality, as introduced in \cite{talamas2017network}, for the transposed interaction matrix; 
\item[(ii)] when $Q=n^{-1}\1\1'$, so that the $Q$-volatility coincides with $n$ times the variance of the average equilibrium value \eqref{variance-average}, we have that 
\be\label{Bonacich-0}\pi_i^2=\frac1n (M'\1\1'M)_{ii}=\frac1n(M'\1)_i^2\,,\ee
so that the $Q$-target centrality of a node 
\be\label{Bonacich-centrality}\pi_i=\frac1{\sqrt{n}}\bigg|\sum_{j=1}^{n}M_{ji}\bigg|\,,\ee
is proportional to its so-called Bonacich centrality \cite{bonacich1987power} when $A$ is entrywise non-negative.
\item[(iii)] when  $Q=I-n^{-1}\1\1'$, so that the $Q$-volatility coincides with the aggregate variance of the displacement from the average equilibrium value \eqref{polarization}, we have that the squared $Q$-target centrality  of a node $i$
\be\label{concentration-centrality}\pi_i^2=\sum_{j=1}^{n}(M_{ji})^2-\frac1n\bigg(\sum_{j=1}^{n}M_{ji}\bigg)^2\,,\ee
is proportional to its so-called  concentration centrality \cite{Acemoglu.ea:16}.
\end{enumerate}

Notice that the $I$-target centrality of all nodes is strictly positive. In contrast, the $n^{-1}\1\1'$-target centrality of some nodes may be $0$ if $A$ has negative entries, and the $I-n^{-1}\1\1'$-target centrality of a node $i$ is $0$ if and only if the $i$-th column of $M$ is constant, or, equivalently, if and only if $(A\1)_i \neq 1$ and $(A\1)_j = 1$ for every $j \neq i$.
The next result relates the system $Q$-volatility to the $Q$-target centrality.
\begin{proposition} \label{prop:volatility-uncorrelated}
Let $A$ in $\R^{n\times n}$ be a Schur stable interaction matrix and let $w_i$, for $i=1,\ldots,n$, be idiosyncratic shocks with variance $\sigma_i^2\ge 0$.  
Then, for every symmetric positive semi-definite matrix $Q$ in $\R^{n\times n}$, the $Q$-volatility of the network equilibrium  \eqref{equilibrium} satisfies
$$\Phi_Q=\sum_{i=1}^{n}\pi_i^2\sigma_i^2\,,$$
where $\pi_i$ is the $Q$-target centrality as defined in \eqref{Q-target-centrality}. 
\end{proposition}
\begin{proof} We have that 
$$\ba{rcl}\ds\Phi_Q
&=&\ds\E\left[w'M'QMw\right]\\[7pt] 
&=&\ds\E\left[\tr(M'QMww')\right]\\[7pt]
&=&\ds\tr\left(M'QM\E\left[ww'\right]\right)\,,\\[7pt]
&=&\tr(M'QM[\sigma^2])\\[7pt]
&=&\ds\sum_{i=1}^{n}(M'QM)_{ii}\sigma_i^2\\[7pt]
&=&\ds\sum_{i=1}^{n}\pi_i^2\sigma_i^2\,,\ea$$
where the first identity follows from \eqref{system-volatility} and  \eqref{equilibrium-bis}, the second one from the fact that $a'b=\tr(ba')$ for every $a$ and $b$ in $\R^n$, the third one from the fact that the matrices $M$ and $Q$ are deterministic, the fourth one from \eqref{noise-idiosync}, and the last one from \eqref{Q-target-centrality}. 
\qed\end{proof}

\begin{example}\label{ex:graph6}
\begin{figure}
\centering
\includegraphics{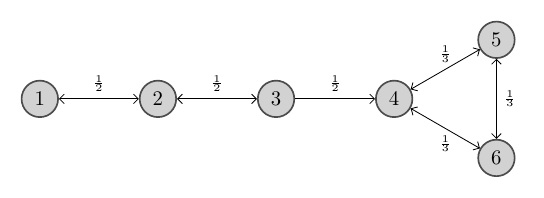}
\vspace{-0.75cm}
\caption{\label{fig:graph6}The network of Example \ref{ex:graph6}.}
\end{figure}
Consider a network with $n=6$ nodes and interaction matrix
\be\label{A6}A=\left(\ba{cccccc}0&1/2&0&0&0&0\\1/2&0&1/2&0&0&0\\0&1/2&0&1/2&0&0\\0&0&0&0&1/3&1/3\\0&0&0&1/3&0&1/3\\0&0&0&1/3&1/3&0\ea\right)\,,\ee
corresponding to the weighted graph displayed in Figure \ref{fig:graph6}. The input-output matrix is then
$$M=(I-A)^{-1}=\left(\ba{cccccc}
3/2&1&1/2&3/8&3/16&3/16\\1&2&1&3/4&3/8&3/8\\1/2&1&3/2&9/8&9/16&9/16
\\0&0&0&3/2&3/4&3/4\\0&0&0&3/4&3/2&3/4\\0&0&0&3/4&3/4&3/2\ea\right)\,.$$
For $Q=I$, we have that 
$$\pi^2=\left(\frac72,6,\frac72,\frac{171}{32},\frac{495}{128},\frac{495}{128}\right)$$
This means that the node whose shock has the highest impact on the total variance of the network equilibrium is $i=2$, followed by nodes $4$, $5$, and $6$. 
On the other hand, for $Q=n^{-1}\1\1'$, we have that $$\pi^2=\frac1{6}\left(9,16,9,\frac{441}{16},\frac{1089}{64},\frac{1089}{64}\right)\,,$$ implying that the shock that has the largest impact on the variance of the aggregate network equilibrium value is $i=4$, followed by nodes $5$ and $6$. 
Finally, for $Q=I-n^{-1}\1\1'$, we have that $$\pi^2=\left(2,\frac{10}3,2,\frac{3}{4},\frac{33}{32},\frac{33}{32}\right)\,,$$ so that the node whose shock has the highest impact on the polarization of the network equilibrium is $i=2$, followed by nodes $1$ and $3$.
\end{example}

\begin{example}\label{ex:star}
\begin{figure}
\centering
\includegraphics{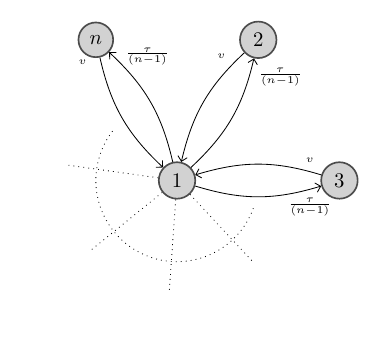}
\vspace{-0.75cm}
\caption{\label{fig:star}The star network of Example \ref{ex:star}.}
\end{figure}
Consider a star network with interaction matrix 
\be\label{Astar}\!\!\!\!A\!=\!\left(\ba{ccccc}0&\tau/(n-1)&\tau/(n-1)&\ldots&\tau/(n-1)\\ 
\ups&0&0&\ldots&0\\
\ups&0&0&\ldots&0\\
\vdots&\vdots&\vdots&&\vdots\\\ 
\ups&0&0&\ldots&0\ea\right)\ee
where $\tau,\ups\in(0,1)$ (see Figure \ref{fig:star}).
Here, node $1$ represents a hub and is connected to all the other nodes, representing spokes, each of which is connected only to the hub.
Setting $\theta=1-\ups\tau$, the input-output matrix is given by 
$$M=\frac{1}{\theta}
\left(\ba{ccccc}1&\frac{\tau}{n-1}&\frac{\tau}{n-1}&\ldots&\frac{\tau}{n-1}\\ 
\ups&\theta+\frac{\ups\tau}{n-1}&\frac{\ups\tau}{n-1}&\ldots&\frac{\ups\tau}{n-1}\\
\ups&\frac{\ups\tau}{n-1}&\theta+\frac{\ups\tau}{n-1}&\ldots&\frac{\ups\tau}{n-1}\\
\vdots&\vdots&\vdots&&\vdots\\\
\ups&\frac{\ups\tau}{n-1}&\frac{\ups\tau}{n-1}&\ldots&\theta+\frac{\ups\tau}{n-1}\ea\right).$$
The $Q$-target centrality vector is then given by: 
\begin{enumerate}
\item[(i)]
for $Q=I$, $$\hspace{-.9cm}\pi_1=\frac{\sqrt{1+(n-1)\ups^2}}{\theta},\quad \pi_{i\ge2}=\frac{\sqrt{1+\theta^2(n-2)+\frac{\tau^2}{n-1}}}{\theta\sqrt{n-1}};$$
\item[(ii)] for $Q=n^{-1}\1\1'$, $$\hspace{-.9cm}\pi_1=\frac{{1+(n-1)\ups}}{\theta\sqrt{n}},\quad \pi_{i\ge2}=\frac{{1+\frac{\tau}{n-1}}}{\theta\sqrt{n}};$$
\item[(iii)] for $Q=I-n^{-1}\1\1'$, $$\hspace{-.9cm}\pi_1=\frac{\sqrt{n-1}(1-\ups)}{\theta\sqrt{n}},\quad \pi_{i\ge2}=\frac{\sqrt{\theta^2(n-2)+\frac{(1-\tau)^2}{n}}}{\theta\sqrt{n-1}}.$$
\end{enumerate}
Now assume, for simplicity of exposition, that $\tau=\ups$.
It is worth observing that, while both the $I$-target and $n^{-1}\1\1'$-target centralities of the hub are larger than the corresponding  ones of the spokes, the  converse holds true for the $(I-n^{-1}\1\1')$-target centrality that is larger for the spokes than for the hub.
\end{example}

\section{Optimal volatility reduction}\label{sec:given-distribution}
In this section, we study the problem faced by a system planner who has the ability to attenuate the effect of some exogenous shocks and aims at designing interventions that optimally reduce the $Q$-volatility of the network.  

Formally, let $z\ge\1$ be an $n$-dimensional shock attenuation vector, to be referred to as an intervention, and let 
\be\label{system} x=Ax+u+[z]^{-1}w\,,\ee
be the network equilibrium model upon intervention.  
We study network intervention problems in which a system planner designs the intervention $z\ge\1$ 
with the objective of minimizing the system volatility \eqref{system-volatility} under a budget constraint 
\be\label{budget-constraint} \sum_{i=1}^{n}\alpha_i^2z_i^2\le\beta\,,\ee 
where $\alpha_i>0$ are cost weights and $\beta$ 
is the available budget. In order to formalize the optimization problem, we define the set of feasible interventions under budget constraint $\beta$ as
\be\label{eq:vincolo}
\mc Z_\beta = \{z \in \R^n_+ \ | \ z \ge \1, \ \sum_{i=1}^n \alpha_i^2 z_i^2 \le \beta\}\,.
\ee
Clearly, a necessary and sufficient condition for this problem to be feasible is that 
\be\label{feasible-budget-threshold}\beta\ge\beta_{\min}\,,\qquad \beta_{\min}=\sum_{i=1}^{n}\alpha_i^2\,.\ee

Given a Schur-stable interaction matrix $A$ and a symmetric positive semi-definite matrix $Q$, we shall denote by $\Phi_Q(z)$ the $Q$-volatility of the network equilibrium model upon intervention \eqref{system}  and  study the optimal volatility reduction problem
\be\label{prob:min}\min_{\ds z \in \mc Z_\beta}\Phi_Q(z)\,.\ee
Notice that, since the intervention consists in a rescaling of each shock component $w_i$ by the quantity $z_i\ge1$, Proposition \ref{prop:volatility-uncorrelated} implies that the $Q$-volatility of the network equilibrium model upon intervention is given by
\be\label{PhiQ-intervention-uncorrelated}\Phi_Q(z)=\sum_{i=1}^{n}\pi_i^2\frac{\sigma_i^2}{z_i^2}\,,\ee
where $\pi_i$ is the $Q$-target centrality of a node $i=1,\ldots,n$, as defined in \eqref{Q-target-centrality}. 
It follows that the optimal $Q$-volatility reduction problem \eqref{prob:min} can be rewritten in the standard form
\be\label{optim}\min_{\substack{z\ge\1:\\\,g(z)\le 0}}f(z)\,,\ee
with $$f(z)=\Phi_Q(z)=\sum_{i=1}^{n}\frac{\sigma_i^2\pi_i^2}{z_i^2}\,,\qquad g(z)=\sum_{i=1}^{n}\alpha_i^2z_i^2-\beta\,.$$
Optimization problems in the form \eqref{optim} can be solved by means of the following standard result. 
\begin{lemma}\label{lemma:simple}
For two continuous functions $f:\R_{++}^n\to\R$ and $g:\R_{++}^n\to\R$, consider the Lagrangian $L(z,\lambda)=f(z)+\lambda g(z)\,.$
Let $z^*\ge\1$ and $\lambda\ge0$ be such that 
\be\label{g=0}g(z^*)=0\,,\ee
and
\be\label{L-opt}f(z^*)=L(z^*,\lambda)=\min_{z\ge\1}L(z,\lambda)\,.\ee
Then, 
\be\label{claim-lemma}f(z^*)=\min_{\substack{z\ge\1:\\\,g(z)\le 0}}f(z)\,.\ee
\end{lemma}
\begin{proof}
It follows from \eqref{g=0} and \eqref{L-opt} that, for every $z\ge\1$ such that $g(z)\le0$, 
$$f(z^*)=L(z^*,\lambda)\le L(z,\lambda)=f(z)+\lambda g(z)\le f(z)\,,$$
thus proving \eqref{claim-lemma}. 
\qed\end{proof}

Now, define the quantities 
\be\label{omegai-def}\omega_i=\sqrt{\frac{\pi_i\sigma_i}{\alpha_i}}\,,\qquad i=1,\ldots,n\,,\ee
and let
$$\omega_{\min}=\min_{i=1,\ldots,n}\omega_i\,,\qquad \omega_{\max}=\max_{i=1,\ldots,n}\omega_i\,.$$
We assume that $\omega_{\max}>0$, otherwise the problem is trivial. 
Notice that the function $\xi:(0,\omega_{\max}]\to[\beta_{\min},+\infty)$ defined by 
\be\label{csi}\xi(t)=\sum_{i=1}^{n}\alpha_i^2\left(1\vee\frac{\omega_i^2}{t^2}\right)\,,\qquad 0<t\le\omega_{\max}\,.\ee
is strictly decreasing with  
\be\label{extremes}\lim_{t\downarrow0}\xi(t)=+\infty\,,\qquad \xi(\omega_{\max})=\beta_{\min}\,.\ee
Hence,  its inverse function
\be\label{zeta-def}\zeta:[\beta_{\min},+\infty)\to(0,\omega_{\max}]\,,\qquad \zeta(\beta)=\xi^{-1}(\beta)\,,\ee
is well defined and strictly decreasing itself.
We can then present the first main result of this paper, explicitly characterizing the optimal volatility reduction intervention. 

\begin{theorem}\label{theo:min-uncor}
Let $A$ in $\R^{n\times n}$ be a Schur stable interaction matrix and let $w_i$, for $i=1,\ldots,n$, be idiosyncratic shocks with variance $\sigma_i^2 \ge 0$.  
Then, for every  symmetric positive semi-definite matrix $Q$ in $\R^{n\times n}$, positive vector $\alpha$  in $\R^n$, and budget $\beta$ satisfying \eqref{feasible-budget-threshold}, the optimal solution $z(\beta)$ in $\R^n$ of the $Q$-volatility reduction problem \eqref{prob:min} is unique and defined by
\be\label{optimal-1} z_i(\beta)=1\vee{\frac{\omega_i}{\zeta(\beta)}}\,,\qquad i=1,2,\ldots,n\,,\ee
where $\omega_i$ and $\zeta(\beta)$ are defined in \eqref{omegai-def} and \eqref{zeta-def}, respectively. 
\end{theorem}
\begin{proof}
Clearly, $z(\beta)\ge\1$. 
Moreover, 
\be\label{constraint=} \ba{rcl}g(z(\beta))&=&\ds\sum_{i=1}^{n}\alpha_i^2z_i^2(\beta)-\beta\\[10pt]
&=&\ds\sum_{i=1}^{n}\alpha_i^2\left(1\vee\frac{\omega_i^2}{\zeta^2(\beta)}\right)-\beta\\[10pt]
&=&\ds\xi\left(\zeta(\beta)\right)-\beta\\[10pt]&=&0\,.\ea\ee
Now, consider the Lagrangian 
$$L(z,\lambda)
=\sum_{i=1}^{n}\frac{\sigma_i^2\pi_i^2}{z_i^2}+\lambda\left(\sum_{i=1}^{n}\alpha_i^2z_i^2-\beta\right)\,,$$
and compute its partial derivatives 
\be\label{derivative-Lagrangian} \ba{rcl}\ds \!\! \frac{\partial L}{\partial z_i}(z(\beta),\zeta^4(\beta))
&=&\ds -2\frac{\pi_i^2\sigma_i^2}{z_i^3(\beta)}+2\zeta^4(\beta)\alpha_i^2z_i(\beta)\\[10pt]
&=&\ds\frac{2\alpha_i^2}{z_i^3(\beta)}\left(\zeta^4(\beta)z_i^4(\beta)-\omega_i^4\right)\,,\ea\ee
for $i=1,2,\ldots,n$. 
Then, notice that, for every $i$ such that $z_i(\beta)=1$, \eqref{optimal-1} implies that $\zeta(\beta)\ge\omega_i$, so that, by \eqref{derivative-Lagrangian} we get
$$\ba{rcl}\ds\frac{\partial L}{\partial z_i}(z(\beta),\zeta^4(\beta))
&=&\ds\frac{2\alpha_i^2}{z_i^3(\beta)}\left(\zeta^4(\beta)z_i^4(\beta)-\omega_i^4\right)\\[10pt]
&=&\ds2\alpha_i^2\left(\zeta^4(\beta)-\omega_i^4\right)\\[8pt]
&\ge&0\,.\ea$$
On the other hand, for every $i$ such that $z_i(\beta)=\omega_i/\zeta(\beta)$, \eqref{optimal-1} implies that  $$\frac{\partial L}{\partial z_i}(z(\beta),\zeta^4(\beta))=\frac{2\alpha_i^2}{z_i^3(\beta)}\left(\zeta^4(\beta)z_i^4(\beta)-\omega_i^4\right)=0\,.$$
Hence, 
\be\label{Lagrangian-min} f(z(\beta))=L(z(\beta),\zeta^4(\beta))=\min_{z\ge\1}L(z,\zeta^4(\beta))\,.\ee
We now apply Lemma \ref{lemma:simple} and get that $z(\beta)$ defined in \eqref{optimal-1} is a solution of \eqref{prob:min}.

To prove uniqueness, let $z^*(\beta)$ be any optimal solution of~\eqref{prob:min} and let
$\mathcal P=\{i:\omega_i>0\}$. Since $\pi_i^2\sigma_i^2=\alpha_i^2\omega_i^4$, we can write
$f(z)=\sum_{i\in\mathcal P}\alpha_i^2\omega_i^4 z_i^{-2}$, where each term is strictly convex
in $z_i>0$. If $z_i^*(\beta)\ne z_i(\beta)$ for some $i\in\mathcal P$, then the point
$\tilde z(\beta)=\frac12\big(z^*(\beta)+z(\beta)\big)$, which belongs to the convex set
$\mathcal Z_\beta$, would satisfy
$f(\tilde z(\beta))<\frac12 f(z^*(\beta))+\frac12 f(z(\beta))=f(z(\beta))$, contradicting the optimality
of $z(\beta)$. Hence, $z_i^*(\beta)=z_i(\beta)$ for every $i\in\mathcal P$. On the other hand,
$z_i(\beta)=1$ for every $i\notin\mathcal P$, so that
$$\begin{aligned}
\sum_{i\notin\mathcal P}\alpha_i^2\big(z_i^{*2}(\beta)-1\big)
& =\sum_{i=1}^n\alpha_i^2\big(z_i^{*2}(\beta)-z_i^2(\beta)\big) \\
&=g(z^*)-g(z(\beta))\le0\,.
\end{aligned}$$
Since $z^*\ge\1$ and $\alpha_i>0$, every term of the sum on the left-hand side is
nonnegative, so that $z_i^*(\beta)=1=z_i(\beta)$ for every $i\notin\mathcal P$. Therefore,
$z^*(\beta)=z(\beta)$.
\qed\end{proof}

Theorem \ref{theo:min-uncor} provides an explicit formula for the optimal volatility reduction problem in case of idiosyncratic shocks.
For every budget level $\beta\ge\beta_{\min}$, we may define the set
\be\label{active}\mc A(\beta)=\l\{i\!=\!1,\!\dots\!,n\!: \omega_i>\zeta(\beta)\r\}\,,\ee
of \emph{active} nodes on which the optimal intervention concentrates.
Then, equation \eqref{optimal-1} may be rewritten as $$z_i(\beta)=\left\{\ba{lcl}\ds{\omega_i}/{\zeta(\beta)}&\se&i\in\mc A(\beta)\\[5pt]1&\se&i\notin\mc A(\beta)\,.\ea\right.$$
The key quantity $\omega_i=\sqrt{\pi_i\sigma_i{/\alpha_i}}$ that determines whether a node $i$ is active ---and, if so, the magnitude of the optimal intervention $z_i(\beta)$ on it--- is the square root of the product of $i$'s $Q$-target centrality $\pi_i$, of the standard deviation $\sigma_i$ of the exogenous shock $w_i$, and of the inverse of the square root $\alpha_i$ of the cost of a unitary intervention. 
There are two extreme scenarios. On the one hand, if $\beta_{\min} < \beta\le\xi(\max\{\omega_i|\omega_i<\omega_{\max}\})$, then only nodes $i$ with the highest value of $\omega_i$ are active, i.e.,  $\mc A(\beta)=\{i:\omega_i=\omega_{\max}\}$;  
on the other hand, if $\omega_{\min} > 0$ and $\beta > \xi(\omega_{\min})$, then all nodes are active, i.e., $\mc A(\beta)=\{1,\ldots,n\}$. Notice that, if $\pi_i = 0$ or $\sigma_i = 0$ for a node $i$, then $\omega_i = 0$ and $z_i(\beta) = 1$ for every budget $\beta$, namely, node $i$ will never be active.

\begin{figure}
\centering
\includegraphics[scale=0.675]{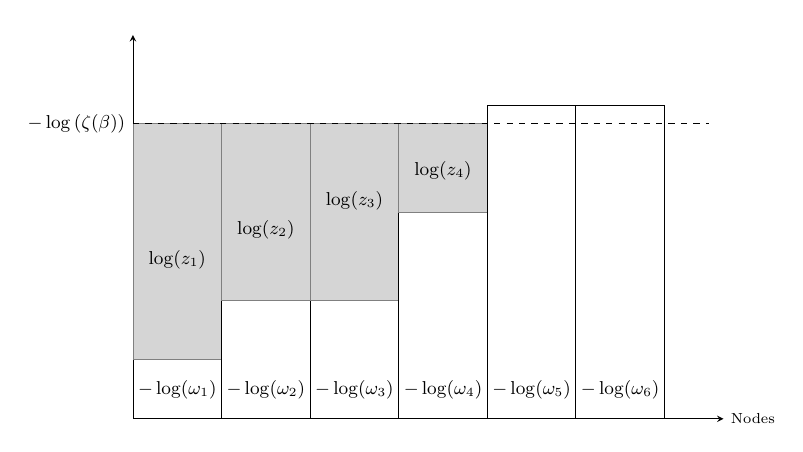}
\caption{Water-filling illustration of the optimal volatility reduction intervention in Example~\ref{ex:graph6}, with nodes ordered so that $\omega_1>\omega_2>\dots>\omega_6$, and $\omega_5<\zeta(\beta)<\omega_4$.}
\label{fig:water-filling}
\end{figure}

For arbitrary values of budget $\beta$, the solution \eqref{optimal-1} of the optimal volatility reduction problem is illustrated graphically in Figure \ref{fig:water-filling} for the network of Example  \ref{ex:graph6}.
The vertical levels indicate the opposite of the logarithm of the parameter $\omega_i$ defined in \eqref{omegai-def}.
As the available budget $\beta$ is increased from its minimum value $\beta_{\min}$, the optimal intervention starts by allotting it to the nodes $i$ with the highest value of $\omega_i$ (equivalently, the lowest value of $-\log\omega_i$). As $\beta$ is further increased, some of the available budget is put on nodes with lower values of $\omega_i$.
The process by which the available budget is allocated to the different nodes resembles the way in which an amount of water distributes itself on an uneven floor composed of $n$ parts of unitary area and height $-\log\omega_i$, for $i=1,\ldots,n$. Having $\omega_i = 0$ is equivalent to a floor of infinite height.
This process is often referred to as \emph{water-filling}, hence we may refer to formula \eqref{optimal-1} as a water-filling-type solution\footnote{Water-filling-type solutions are common in certain problems of Information Theory, such as power allocation in parallel Gaussian transmission channels \cite[Ch.~9.4]{Cover.Thomas:06}.} of the optimal volatility reduction problem \eqref{prob:min}.

Theorem \ref{theo:min-uncor} also allows us to compute the minimal $Q$-volatility $\Phi_{Q}\l(z(\beta)\r)$ and study its dependence on budget $\beta$.
In particular, we have the following result. 
\begin{corollary}\label{coro:min-uncor-PhiQ}
Let $A$ in $\R^{n\times n}$ be a Schur stable interaction matrix and let $w_i$, for $i=1,\ldots,n$, be idiosyncratic shocks with variance $\sigma_i^2 \ge 0$.  
Then, for every  symmetric positive semi-definite matrix $Q$ in $\R^{n\times n}$, positive vector $\alpha$  in $\R^n$, and budget $\beta > \beta_{\min}$, the minimal $Q$-volatility achievable in the optimization problem \eqref{prob:min} is
\be\label{optimal-1-PhiQ-active}\Phi_{Q}\l(z(\beta)\r)=\frac{\Big(\sum\limits_{i\in\mc A(\beta)}\alpha_i\pi_i\sigma_i\Big)^2}{\beta-\sum\limits_{i\notin\mc A(\beta)}\alpha_i^2}+\sum_{i\notin\mc A(\beta)}\pi_i^2\sigma_i^2\,,\ee
where $\pi_i$ is the $Q$-target centrality of node $i$, as defined in \eqref{Q-target-centrality}, and $\mc A(\beta)$ is the set of active nodes defined in \eqref{active}. 
\end{corollary}
\begin{proof}
By substituting formula \eqref{optimal-1} into \eqref{PhiQ-intervention-uncorrelated} and using \eqref{omegai-def}, we get 
\be\label{phi-special2}\ba{rcl}\ds\Phi_Q(z(\beta))
&=&\ds\sum_{i=1}^n\frac{\pi_i^2\sigma_i^2}{1\vee\frac{\omega_i^2}{\zeta^2(\beta)}}\\[10pt]
&=&\ds\sum_{i=1}^n\pi_i^2\sigma_i^2\l(1\wedge\frac{\zeta^2(\beta)}{\omega_i^2}\r)\\[10pt]
&=&\ds\sum_{i \in \mc A(\beta)} \frac{\pi_i^2 \sigma_i^2 \zeta^2(\beta)}{\omega_i^2} + \sum_{i \notin \mc A(\beta)} \pi_i^2 \sigma_i^2\\[10pt]
&=&\ds\zeta^2(\beta)\sum\limits_{i\in\mc A(\beta)}\alpha_i\pi_i\sigma_i+\sum\limits_{i\not\in\mc A(\beta)}\pi_i^2\sigma_i^2\,.\ea
\ee
Notice now that for $t=\zeta(\beta)$, we have that $$\beta=\xi(t)=\sum\limits_i\alpha_i^2\l(1\vee\frac{\omega_i^2}{t^2}\r)=\frac{1}{t^2}\sum\limits_{i\in\mc A(\beta)}\alpha_i\pi_i\sigma_i+\sum\limits_{i\not\in\mc A(\beta)}\alpha_i^2\,,$$
so that $$\zeta^2(\beta)=\frac{\sum\limits_{i\in\mc A(\beta)}\alpha_i\pi_i\sigma_i}{\beta- \sum\limits_{i\not\in\mc A(\beta)}\alpha_i^2}\,.$$
Substituting the above into equation \eqref{phi-special2}, we obtain \eqref{optimal-1-PhiQ-active}.
\qed\end{proof}

Corollary \ref{coro:min-uncor-PhiQ} characterizes the minimum achievable $Q$-volatility as a function of the budget level $\beta$, of the $Q$-target centrality vector $\pi$, of the shocks' variances $\sigma_i^2$, and of the cost of a unitary intervention $\alpha_i^2$ in the nodes $i=1,\ldots,n$. 
Notice that if $\omega_{\min} > 0$ and $\beta\ge\xi(\omega_{\min})$, then \eqref{optimal-1-PhiQ-active} reduces to
\be\label{optimal-1-PhiQ-HB}\Phi_Q(z(\beta))=\frac{1}{\beta}\l(\sum_{i=1}^n\alpha_i\pi_i\sigma_i\r)^2\,.\ee
Thus, unless some nodes have $\omega_i = 0$, the optimal $Q$-volatility is inversely proportional to the budget level $\beta$ whenever $\beta$ is sufficiently large that all nodes are active in the optimal intervention. Instead, for lower budget levels, where the optimal intervention is concentrated on a proper subset of active nodes $\mc A(\beta)$, it is only the global effect of the exogenous shocks on the nodes in $\mc A(\beta)$ that gets reduced by the budget level $\beta$, while the contribution of the other nodes on the minimal $Q$-volatility remains constant.

\addtocounter{example}{-2}	
\begin{example}[continued]
Consider the network with $n=6$ nodes and interaction matrix \eqref{A6}. 
Assume that $\alpha_i^2=1$ and $\sigma_i^2=1$, for $i=1,2,\dots,n$. Then, for every $i$, $\omega_i=\sqrt{\pi_i}$. Denote by $z_Q(\beta)\in\R^n$ the optimal intervention with respect to $Q$-volatility and set $J=n^{-1}\1\1'$.
Figure \ref{fig:graph6-z} depicts the behavior of $z_Q(\beta)$ when $Q=I$ (solid line), $Q=J$ (dashed line), and $Q=I-J$ (dotted line).
\begin{figure}
\centering
\includegraphics{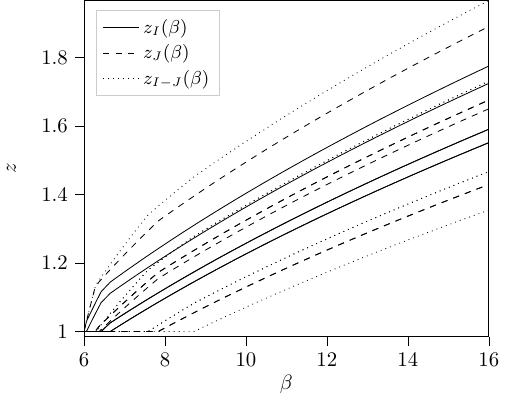}
\caption{\label{fig:graph6-z} Behavior of $z(\beta)$ for network of Example \ref{ex:graph6} when $Q$ is equal to $I,J=n^{-1}\1\1'$, and $I-J$. 
}
\end{figure}
\end{example}

\begin{example}[continued]
\label{ex:star-1}
Consider a star network with interaction matrix \eqref{Astar}.
Assume that $\alpha_i^2=1$ and $\sigma_i^2=1$, for $i=1,2,\dots,n$ and that $\tau=\ups\in(0,1)$. We set $J=n^{-1}\1\1'$ and $\theta=1-\tau^2$.
We now compute the optimal $Q$-volatility using Corollary \ref{coro:min-uncor-PhiQ}.

(i) When $Q=I$, $\pi_1\ge\pi_{i\ge2}$.
The hub is targeted first and, subsequently, for large enough $\beta$, the intervention becomes active also on all spokes. 
The transition occurs at $$\beta_I=n-1+\sqrt{\frac{(n-1)(1+(n-1)\tau^2)}{1+\theta^2(n-2)+\frac{\tau^2}{n-1}}}\,.$$
Then, the resulting optimal $Q$-volatility is $$\Phi_Q(z(\beta))=\l\{\ba{ll} n-2+\frac{n\tau^2+\theta}{\theta^2(\beta-(n-1))}+\frac{n-\theta}{\theta^2(n-1)} & \mbox{ if } \beta<\beta_I\\&\\ \frac{\l(\sqrt{n\tau^2+\theta}+\sqrt{n-\theta+\theta^2(n-1)(n-2)}\r)^2}{\theta^2\beta} & \mbox{ if } \beta\ge\beta_I.\ea\r.$$

(ii) When $Q=J$, $\pi_1\ge\pi_{i\ge2}$.
The hub is the first node to be targeted and, then, for large enough $\beta$, the intervention becomes active also on all spokes.
The transition occurs at $$\beta_{J}=n-1+\frac{1+(n-1)\tau}{1+\frac{\tau}{n-1}}\,.$$
Then, the resulting optimal $Q$-volatility is
$$\Phi_Q(z(\beta))=\l\{\ba{ll} \frac{(1+(n-1)\tau)^2}{n\theta^2(\beta-(n-1))}+\frac{(n-1+\tau)^2}{n(n-1)\theta^2} & \mbox{ if } \beta<\beta_{J}\\&\\ \frac{n}{(1-\tau)^2\beta} & \mbox{ if } \beta\ge\beta_{J}.\ea\r.$$

(iii) When $Q=I-J$, $\pi_{i\ge2}\ge\pi_1$.
Spokes are the first nodes to be targeted and, then, also the intervention on the hub becomes active.
The transition occurs at $$\beta_{I-J}=1+\sqrt{1+n(n-2)(1+\tau)^2}\,.$$ 
Then, the resulting optimal $Q$-volatility is $$\Phi_Q(z(\beta)) \! = \! \l\{\ba{ll} \!\!\!\! \frac{n-1}{n(1+\tau)^2}+\frac{(n-1)\big(1+n(n-2)(1+\tau)^2}{n(1+\tau)^2(\beta-1)\big)} & \!\!\!\! \mbox{ if } \beta<\beta_{I-J}\\&\\ \!\!\!\! \frac{n-1}{n\beta}\l(\frac{1+\sqrt{1+n(n-2)(1+\tau)^2}}{1+\tau}\r)^2 & \!\!\!\! \mbox{ if } \beta\ge\beta_{I-J}.\ea\r.$$

Assuming $n=7$ and $\tau=0.5$, Figure \ref{fig:star-phi} shows the behavior of $\Phi_Q(z(\beta))$ when $Q=I$ (solid line), $Q=J$ (dashed line), and $Q=I-J$ (dotted line). Vertical dashed lines represent values of $\beta_Q$.
It is possible to prove that $\beta_I\le\beta_J\le\beta_{I-J}$ for all $\tau\in(0,1)$ and all $n\ge3$.
\begin{figure}
\centering
\includegraphics{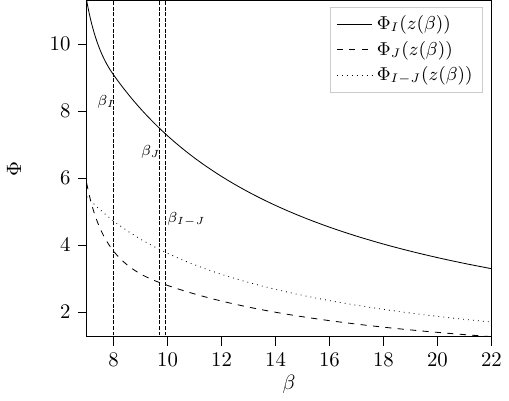}
\caption{\label{fig:star-phi} Behavior of $\Phi_Q(z(\beta))$ for star network of Example \ref{ex:star-1} with $n=7$ and $\tau=\ups=0.5$, when $Q$ is equal to $I,J=n^{-1}\1\1'$, and $I-J$.}
\end{figure}
\end{example}

\section{Robust volatility reduction}\label{sec:min-max}
In this section, we study a robust version of the optimal volatility reduction problem where the system planner is interested in designing a shock attenuation vector $z$ that minimizes the system volatility against arbitrary idiosyncratic shocks $w_i$ whose variances $\sigma_i^2 \ge 0$ satisfy
\be\label{sigma-constraint}\sum_{i=1}^{n}\eta_i^2\sigma_i^2\le 1\,,\ee
where $\eta_i>0$, $i=1,2,\ldots,n$ are given weights.

Formally, for a symmetric positive semi-definite matrix $Q$, we shall solve the min-max problem 
\be\label{prob:min-max}\min_{\substack{\ds z \in \mc Z_\beta}}\ \max_{\substack{\ds\sigma \geq 0:\ds\eqref{sigma-constraint}}}\Phi_Q(z,\sigma)\,,\ee
where the notation $\Phi_Q(z,\sigma)$ is introduced in order to explicitly indicate the dependence of the $Q$-volatility of the network equilibrium on both the intervention $z$ and the vector $\sigma$ of the idiosyncratic shocks. 

Equation \eqref{prob:min-max} may be interpreted as an adversarial min-max optimization problem. First, the system planner chooses an intervention $z$ in $\mc Z_\beta$ in order to minimize the $Q$-volatility of the network equilibrium \eqref{system} in the worst case scenario, then the adversary chooses the vector  $\sigma$ of the standard deviations of the idiosyncratic shocks within the constraint \eqref{sigma-constraint} in order to maximize the $Q$-volatility of the network equilibrium itself.

\begin{remark} It follows from Proposition \ref{prop:volatility-uncorrelated} that the robust optimal volatility reduction problem \eqref{prob:min-max} can be rewritten as 
\be\label{robust-rewritten}\min_{\ds z \in \mc Z_\beta}\ \max_{\substack{\ds\sigma\ge0:\eqref{sigma-constraint}}}\sum_{i=1}^n\pi_i^2\frac{\sigma_i^2}{z_i^2}\,.\ee
Eq. \eqref{robust-rewritten} can be interpreted as the value of a 2-player zero-sum strategic game where the players are the system planner, choosing action $a=z^2$ from $\mc A=\{a\in\R^n:\,a\ge\1,\,\sum_i\alpha_i^2a_i\le\beta\}$, and the adversary, choosing action $b=\sigma^2$ from the set $\mc B=\{b\in\R^n:\,b\ge0,\,\sum_i\eta_i^2b_i\le1\}$, while the utility functions are $$u_1(a,b) =-\sum_{i=1}^n\pi_i^2\frac{b_i}{a_i}\,,\qquad u_2(b,a)=\sum_{i=1}^n\pi_i^2\frac{b_i}{a_i}\,.$$
Since both sets $\mc A$ and $\mc B$ are nonempty, compact, and convex, and $u_1(a,b)=-u_2(b,a)$ is continuous,  concave in $a$ for every $b$ in $\mc B$, and convex in $b$ for every $a$ in $\mc A$, a classical result in game theory \cite[Theorem 1.2]{Fudenberg.Tirole:91} implies  that there exists a Nash equilibrium $(a^*,b^*)$, and that  for every Nash equilibrium $$\min_{a\in\mc A}\max_{b\in\mc B}\sum_{i=1}^n\pi_i^2\frac{b_i}{a_i}=\sum_{i=1}^n\pi_i^2\frac{b_i^*}{a_i^*}=\max_{b\in\mc B}\min_{a\in\mc A}\sum_{i=1}^n\pi_i^2\frac{b_i}{a_i}\,.$$
Hence, $$\min_{\ds z \in \mc Z_\beta} \ \max_{\substack{\ds\sigma\ge0:\eqref{sigma-constraint}}}\!\! \Phi_Q(z,\sigma)\!=\!\!
\max_{\substack{\ds\sigma\ge0:\eqref{sigma-constraint}}}\ \min_{\ds z \in \mc Z_\beta} \! \Phi_Q(z,\sigma)\,.$$
\end{remark}

The following simple result provides an explicit solution to the inner maximization in \eqref{prob:min-max}. 
\begin{lemma}\label{lemma:max}
Let $A$ in $\R^{n\times n}$ be a Schur stable interaction matrix and let $w_i$, for $i=1,\ldots,n$, be idiosyncratic shocks.  
Then, for every  symmetric positive semi-definite matrix $Q$ in $\R^{n\times n}$,  positive vector $\eta$ in $\R^n$ and intervention $z\ge\1$ in $\R^n$, $$\max_{\substack{\ds\sigma\ge0:\ds\eqref{sigma-constraint}}}\Phi_Q(z,\sigma)=\max_{i=1,\ldots,n}\frac{\pi_i^2}{z_i^2\eta_i^2}\,,$$ where $\pi_i$ are the $Q$-target centralities of the nodes $i=1,\ldots,n$, as defined in \eqref{Q-target-centrality}.
\end{lemma}
\begin{proof}
It follows from \eqref{PhiQ-intervention-uncorrelated} that
\be\label{inequalities}\ba{rcl}
\Phi_Q(z,\sigma)
&=&\ds\sum_{i=1}^{n}\frac{\sigma_i^2\pi_i^2}{z_i^2}\\[15pt]
&=&\ds\sum_{i=1}^{n}\frac{\pi_i^2}{z_i^2\eta_i^2}\eta_i^2\sigma_i^2\\[15pt]
&\le&\ds\max_{1\le i\le n}\frac{\pi_i^2}{z_i^2\eta_i^2}\sum_{i=1}^n\eta_i^2\sigma_i^2\\[15pt]
&\le&\ds\max_{1\le i\le n}\frac{\pi_i^2}{z_i^2\eta_i^2}\,,\ea\ee
where the last step follows from \eqref{sigma-constraint}. 
Notice that both inequalities in \eqref{inequalities} hold true as equalities for $\sigma=\sigma^*$ with  
$$\sigma_i^*=\l\{\ba{lcl}1/\eta_i&\se&i=i^*\\0&\se&i\ne i^*\,,\ea\r.$$
for some $i^*$ such that $i^* \in \arg\max\{\pi_{i}/(z_i\eta_i), i = 1,\cdots,n\}$. 
Since $\sum_{i=1}^{n}(\eta_i\sigma^*_i)^2=1$, we have that $\sigma^*$ satisfies \eqref{sigma-constraint}. Therefore, we have shown that  
$$\ba{rcl}\ds\max_{1\le i\le n}\frac{\pi_i^2}{z_i^2\eta_i^2}
&=&\sum_{i=1}^{n}\frac{(\sigma_i^*)^2\pi_i^2}{z_i^2}\\
&\le&\ds\max_{\substack{\ds\sigma\ge0:\ds\eqref{sigma-constraint}}}\Phi_Q(z,\sigma)\\
&\le&\ds\max_{1\le i\le n}\frac{\pi_i^2}{z_i^2\eta_i^2}\,,\ea$$
which readily implies the claim. 
\qed\end{proof}

We now define
\be\label{ovomegai-def}\bar\omega_i=\frac{\pi_i}{\eta_i}\,,\qquad i=1,\ldots,n\,,\ee
and let
$$\bar\omega_{\min}=\min_{i=1,\ldots,n}\,\bar\omega_i\,,\qquad \bar\omega_{\max}=\max_{i=1,\ldots,n}\,\bar\omega_i\,.$$
We assume that $\bar\omega_{\max}>0$, otherwise the problem is trivial. 
Notice that the function $\bar\xi:(0,\bar\omega_{\max}]\to[\beta_{\min},+\infty)$ defined by $$\bar\xi(t)=\sum_{i=1}^{n}\alpha_i^2\l(1\vee\frac{\bar\omega_i^2}{t^2}\r)\,,$$ is strictly decreasing with $$\lim_{t\downarrow0}\bar\xi(t)=+\infty\,,\qquad \bar\xi(\bar\omega_{\max})=\beta_{\min}\,.$$
Hence, its inverse function 
\be\label{ovzeta-def}\bar\zeta:[\beta_{\min},+\infty)\to(0,\bar\omega_{\max}]\,,\qquad \bar\zeta(\beta)=\bar\xi^{-1}(\beta)\,,\ee
is well defined and strictly decreasing itself.

The following is the second main result of the paper, providing an explicit characterization of the optimal robust volatility reduction intervention in case of idiosyncratic shocks. 
\begin{theorem}\label{theo:minmax-uncor}
Let $A$ in $\R^{n\times n}$ be a Schur stable interaction matrix and let $w_i$, for $i=1,\ldots,n$, be idiosyncratic shocks.  
Then, for every  symmetric positive semi-definite matrix $Q$ in $\R^{n\times n}$, positive vectors $\alpha$ and $\eta$  in $\R^n$, and budget $\beta$ satisfying \eqref{feasible-budget-threshold}, 
the optimal solution $\bar z(\beta)$ in $\R^n$ of the $Q$-volatility reduction problem \eqref{prob:min-max} is unique and defined by
\be\label{optimal-2} \bar z_i(\beta)=1\vee{\frac{\bar\omega_i}{\bar \zeta(\beta)}}\,,\qquad i=1,2,\ldots,n\,,\ee
where $\bar\omega_i$ and $\bar\zeta(\beta)$ are defined in \eqref{ovomegai-def} and \eqref{ovzeta-def}, respectively. 
\end{theorem}
\begin{proof} 
It follows from Lemma \ref{lemma:max} that the  robust $Q$-volatility reduction problem \eqref{prob:min-max} reduces to the optimization problem \eqref{optim} with 
$$f(z)=\max_{1\le i\le n}\frac{\pi_i^2}{\eta_i^2z_i^2}=\max_{1\le i\le n}\frac{\bar\omega_i^2}{z_i^2}\,,\qquad g(z)=\sum_{i=1}^{n}\alpha_i^2z_i^2-\beta\,.$$ 
Clearly, $\bar z(\beta)\ge\1$.
Moreover, 
$$\ba{rcl}\ds g(\bar z(\beta))
&=&\ds\sum_{i=1}^{n}\alpha_i^2\bar z_i^2(\beta)-\beta\\[15pt]
&=&\ds\sum_{i=1}^{n}\alpha_i^2\left(1\vee\frac{\bar\omega_i^2}{\bar \zeta^2(\beta)}\right)-\beta\\[15pt]
&=&\ds\bar\xi(\bar\zeta(\beta))-\beta\\[5pt]&=&0\,.\ea$$

Now, let $$a(\beta)=\left(\sum_{i:\bar\omega_i\ge\bar\zeta(\beta)}\alpha_i^2\frac{\bar\omega_i^2}{\bar\zeta^4(\beta)}\right)^{-1}\,,$$ $$b(\beta)=\left(\sum_{i:\bar\omega_i>\bar\zeta(\beta)}\alpha_i^2\frac{\bar\omega_i^2}{\bar\zeta^4(\beta)}\right)^{-1}\,,$$
with $b(\beta) = +\infty$ if the sum is empty. Since $\bar \zeta(\beta) \le \bar \omega_{\max}$, the first sum is non-empty, so that $0 < a(\beta) \le b(\beta)$. 
Define the function $h:(0,\bar\omega_{\max}^2]\to\R$ $$h(\tau)=\tau+\lambda\left(\bar\xi(\sqrt\tau)-\beta\right)\,.$$
Observe that $h(\tau)$ is a continuous, piecewise differentiable, and convex function of $\tau\in(0,\bar\omega_{\max}^2]$ whose left and right derivatives are given by $$h'_-(\tau)=1-\lambda\sum_{i:\bar\omega_i^2\ge\tau}\alpha_i^2\frac{\bar\omega_i^2}{\tau^2}\,,\qquad h'_+(\tau)=1-\lambda\sum_{i:\bar\omega_i^2>\tau}\alpha_i^2\frac{\bar\omega_i^2}{\tau^2}\,.$$
Then, for every $\lambda$ in $[a(\beta),b(\beta)]$, we have that 
$$\ba{rcl}\ds h'_-\left(\bar\zeta^2(\beta)\right)
&=&\ds1-\lambda\sum_{i:\bar\omega_i\ge\bar\zeta(\beta)}\alpha_i^2\frac{\bar\omega_i^2}{\bar\zeta^4(\beta)}\\[15pt]
&\le&0\\
&\le&\ds 1-\lambda\sum_{i:\bar\omega_i>\bar\zeta(\beta)}\alpha_i^2\frac{\bar\omega_i^2}{\bar\zeta^4(\beta)}\\[15pt]
& = &\ds h'_+\left(\bar\zeta^2(\beta)\right)\,,\ea$$
so that $h(\tau)$ attains its minimum at $\tau=\bar\zeta^2(\beta)$ with 
\be\label{a1} \ba{rcl}\ds\min_{\tau\in(0,\bar\omega_{\max}^2]}h(\tau)
&=&\ds h\left(\bar\zeta^2(\beta)\right)\\
&=&\bar\zeta^2(\beta)+\lambda\left(\bar\xi(\bar\zeta(\beta))-\beta\right)\\[4pt]
&=&\bar\zeta^2(\beta)\,.\ea\ee
Moreover, since $\bar\zeta(\beta)\le\bar\omega_{\max}$,
\be\label{eq:cost_robust}
f(\bar z(\beta))\!=\!\max_{i}\frac{\bar\omega_i^2}{1\vee\frac{\bar\omega_i^2}{\bar\zeta^2(\beta)}}
\!=\!\max_{i}\big(\bar\omega_i^2\wedge\bar\zeta^2(\beta)\big)\!=\!\bar\zeta^2(\beta)\,.
\ee
Now, observe that, for every $\tau\in(0,\bar\omega_{\max}^2]$ and every $z$ in $\R^n$ such that $z\ge\1$ and $f(z)=\max_i{\bar\omega_i^2}/{z_i^2}=\tau$, we have that 
\be\label{a2} \ba{rcl}g(z)
&=&\ds\sum_{i=1}^{n}\alpha_i^2z_i^2-\beta\\[15pt]
&\ge&\ds \sum_{i=1}^{n}\alpha_i^2\left(1\vee\frac{\bar\omega_i^2}{\tau}\right)-\beta\\[15pt]
&=&\ds\bar\xi(\sqrt\tau)-\beta\,,
\ea\ee
so that, for every $\lambda\ge0$,
$$
L(z,\lambda)=f(z)+\lambda g(z)\ \ge\ \tau+\lambda\big(\bar\xi(\sqrt\tau)-\beta\big)=h(\tau)\,.$$
It then follows from \eqref{a1} and \eqref{a2} that, for every $z$ in $\R^n$ such that $z\ge\1$ and $\lambda$ in $[a(\beta),b(\beta)]$, we have $$L(z,\lambda)=f(z)+\lambda g(z)\ge \min_{\tau\in(0,\bar\omega_{\max}^2]}h(\tau)=\bar\zeta^2(\beta)=f(\bar z(\beta))\,.$$
We now apply Lemma \ref{lemma:simple} and get that $\bar z(\beta)$ defined in \eqref{optimal-2} is an optimal solution of \eqref{prob:min-max}. Finally, let $z^*$ be any optimal solution. Then $f(z^*)=f(\bar z(\beta))=\bar\zeta^2(\beta)$,
i.e., by the definition of $f$, $z_i^*\ge\bar\omega_i/\bar\zeta(\beta)$ for every $i$, so that
$z_i^*\ge1\vee\bar\omega_i/\bar\zeta(\beta)=\bar z_i(\beta)$. Since
$\sum_i\alpha_i^2\big(z_i^{*2}-\bar z_i^2(\beta)\big)=g(z^*)-g(\bar z(\beta))\le0$ and
$\alpha_i>0$, it follows that $z^*=\bar z(\beta)$.
\qed\end{proof}

Theorem \ref{theo:minmax-uncor}  provides an explicit formula for the optimal solution of the robust volatility reduction problem in case of idiosyncratic shocks. 
As in the non-robust case \eqref{optimal-1}, solution \eqref{optimal-2} has a water-filling-type structure.
In fact, the key difference between \eqref{optimal-2} and \eqref{optimal-1} is that the role played by the quantity $\omega_i=\sqrt{\pi_i\sigma_i{/\alpha_i}}$ in \eqref{optimal-1} is now played in \eqref{optimal-2}  by $\bar\omega_i=\pi_i/\eta_i$ that is the ratio between the $Q$-target centrality of node $i$ and the square root of the cost for the adversary to induce a shock of unitary variance in $i$. Note, in particular, that the sequence of activations in the robust volatility reduction problem does not depend on the intervention cost $\alpha$, since $\bar \omega$ is independent of $\alpha$.
Similarly to our previous discussion following Theorem \ref{theo:min-uncor}, for every budget level $\beta\ge\beta_{\min}$, we may rewrite \eqref{optimal-2} as $$\bar z_i(\beta)=\l\{\ba{lcl}\ds{\bar\omega_i}/{\bar\zeta(\beta)}&\se&i\in\bar{\mc A}(\beta)\\[10pt]1&\se&i\notin\bar{\mc A}(\beta)\,,\ea\r.$$ where $$\bar{\mc A}(\beta)=
\left\{i\!=\!1,\!\dots\!,n\!: \bar\omega_i>\bar\zeta(\beta)\right\}\,,$$ is the set of active nodes on which the optimal robust intervention concentrates its effort. 
We may also observe that if $\beta_{\min} < \beta \le \bar\xi(\max\{\bar\omega_i|\bar\omega_i<\bar\omega_{\max}\})$, then the active nodes for the optimal robust volatility reduction intervention are $\bar{\mc A}(\beta)=\{i:\bar\omega_i=\bar\omega_{\max}\}$, whereas if $\bar\omega_{\min} > 0$ and $\beta>\bar\xi(\bar\omega_{\min})$, then $\bar{\mc A}(\beta)=\{1,\ldots,n\}$, i.e., the  optimal robust volatility reduction intervention effectively acts on every node.

As a consequence of Theorem \ref{theo:minmax-uncor}, we can derive the following result characterizing the value of the minimal robust $Q$-volatility. To this end, let $\bar\Phi_Q(\bar z(\beta)) = \max_{\sigma}\Phi_Q(\bar z(\beta),\sigma)$.

\begin{corollary}\label{coro:minmax-uncor-PhiQ}
Let $A$ in $\R^{n\times n}$ be a Schur stable interaction matrix and let $w_i$, for $i=1,\ldots,n$, be idiosyncratic shocks.  
Then, for every  symmetric positive semi-definite matrix $Q$ in $\R^{n\times n}$, positive vectors $\alpha$ and $\eta$  in $\R^n$, and budget $\beta \ge \beta_{\min}$, the minimal $Q$-volatility achievable in the robust optimization problem \eqref{prob:min-max} is 
\be\label{optimal-2-PhiQ} \bar\Phi_Q(\bar z(\beta))
=\bar\zeta^2(\beta)\,,\ee
where $\bar\omega_i$ and $\bar\zeta(\beta)$ are defined in \eqref{ovomegai-def} and \eqref{ovzeta-def}, respectively. 
\end{corollary}
\begin{proof} 
	It follows from Lemma~\ref{lemma:max} that
	$\max_{\sigma}\Phi_Q(\bar z(\beta),\sigma)=f(\bar z(\beta)) = \bar\zeta^2(\beta)$,
	where the last equality has been shown in the proof of Theorem~\ref{theo:minmax-uncor}, see \eqref{eq:cost_robust}.
\qed\end{proof}

Corollary \ref{coro:minmax-uncor-PhiQ} provides an explicit expression for the minimal achievable value $\bar\Phi_{Q}\l(\bar z(\beta)\r)$ in the robust $Q$-volatility reduction problem \eqref{prob:min-max}.
In particular, equation \eqref{optimal-2-PhiQ} expresses the dependence of the minimal robust $Q$-volatility $\bar\Phi_{Q}(\bar z(\beta))$ on the budget $\beta$, on the $Q$-target centrality vector $\pi$, and on the cost $\eta_i^2$ for the adversary to induce a shock of unitary variance in the nodes $i=1,\ldots,n$. 
Notice that, similarly to the non-robust case, if $\bar\omega_{\min} > 0$ and $\beta\ge\bar\xi(\bar\omega_{\min})$, then \eqref{optimal-2-PhiQ} reduces to
\be\label{optimal-2-PhiQ-HB}\bar\Phi_Q(\bar z(\beta))=\bar\zeta^2(\beta)=\frac{1}{\beta}\sum_{i=1}^n\alpha_i^2\bar\omega_i^2\,.\ee
This shows that for sufficiently high budget levels, where all nodes are active in the robust optimal intervention, the minimal robust $Q$-volatility $\Phi_{Q}(\bar z(\beta))$ is inversely proportional to the budget level $\beta$. 

In fact, it follows from Lemma \ref{lemma:max} that the objective of the system planner when designing the optimal robust intervention $\bar z(\beta)$ is to minimize $$\max_{i=1,\ldots,n}\pi_i/(\eta_i z_i)=\max_{i=1,\ldots,n}\bar\omega_i/z_i\,,$$ over all feasible interventions $z\ge\1$ satisfying the budget constraint \eqref{budget-constraint}.
When $\beta\ge\bar\xi(\bar\omega_{\min})$, the available budget is sufficient to allow for an optimal intervention  $\bar z(\beta)$ that is able to equalize all the ratios $\bar\omega_i/\bar z_i(\beta)$: as a result, in this case the most severe idiosyncratic shocks satisfying the constraint \eqref{sigma-constraint} can induce a $Q$-volatility of at most $\bar\zeta^2(\beta)=\frac{1}{\beta}\sum_{i=1}^n\alpha_i^2\bar\omega_i^2$.
On the other hand, when the budget level is such that  $\beta<\bar\xi(\bar\omega_{\min})$, no feasible intervention  $z$ is able to equalize all the ratios $\bar\omega_i/z_i=\pi_i/(\eta_iz_i)$.
A worst-case idiosyncratic shock is then concentrated in any node $i$ that maximizes $\bar\omega_i/z_i$ and the resulting $Q$-volatility is $\bar\zeta^2(\beta)>\frac{1}{\beta}\sum_{i=1}^n\alpha_i^2\bar\omega_i^2$.

\addtocounter{example}{-2}	
\begin{example}[continued]
Consider the network with $n=6$ nodes and interaction matrix \eqref{A6}. 
Assume that $\alpha_i^2=1$ and $\eta_i^2=1$, for $i=1,2,\dots,n$.
Then, for every $i$, $\bar\omega_i=\pi_i$, the $Q$-target centrality of the node.
Denote by $\bar z_Q(\beta)\in\R^n$ the optimal intervention \eqref{optimal-2} with respect to $Q$-volatility and set $J=n^{-1}\1\1'$.

Figure \ref{fig:graph6-z-bar} depicts the behavior of $\bar z_Q(\beta)$ when $Q=I$ (solid line), $Q=J$ (dashed line), and $Q=I-J$ (dotted line).

\begin{figure}
\centering
\includegraphics{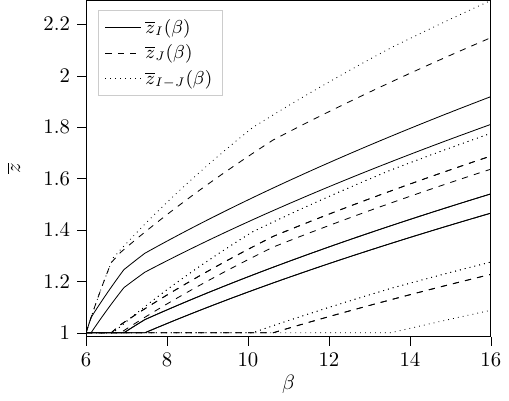}
\caption{\label{fig:graph6-z-bar} Behavior of $\bar z(\beta)$ for network of Example \ref{ex:graph6} when $Q$ is equal to $I,J=n^{-1}\1\1'$, and $I-J$.}
\end{figure}
\end{example}

\begin{example}[continued]
\label{ex:star-2}
Consider a star network with interaction matrix \eqref{Astar}.
Assume that $\alpha_i^2=1$ and $\eta_i^2=1$, for $i=1,2,\dots,n$ and that $\tau=\ups\in(0,1)$. We set $J=n^{-1}\1\1'$ and $\theta=1-\tau^2$.
We now compute the optimal $Q$-volatility using Corollary \ref{coro:minmax-uncor-PhiQ}.

(i) When $Q=I$, $\pi_1\ge\pi_{i\ge2}$.
The transition occurs at $$\bar\beta_I=n-1+\frac{(n-1)(1+(n-1)\tau^2)}{1+\theta^2(n-2)+\frac{\tau^2}{n-1}}\,.$$
Then, the resulting optimal $Q$-volatility is $$\bar\Phi_Q(\bar z(\beta))=\l\{\ba{ll} \frac{n\tau^2+\theta}{\theta^2(\beta-(n-1))} & \mbox{ if } \beta<\bar\beta_I\\&\\ \frac{n-2}{\beta} + \frac{2+(n-1)\tau^2+\frac{\tau^2}{n-1}}{\theta^2\,\beta} & \mbox{ if } \beta\ge\bar\beta_I.\ea\r.$$

(ii) When $Q=J$, $\pi_1\ge\pi_{i\ge2}$.
The transition occurs at $$\bar\beta_{J}=n-1+\l(\frac{1+(n-1)\tau}{1+\frac{\tau}{n-1}}\r)^2\,.$$
Then, the resulting optimal $Q$-volatility is $$\bar\Phi_Q(\bar z(\beta))=\l\{\ba{ll} \frac{(1+(n-1)\tau)^2}{n\theta^2(\beta-(n-1))} & \mbox{ if }\beta<\bar\beta_{J}\\&\\ \frac{(1+(n-1)\tau)^2 + \frac{(n-1+\tau)^2}{n-1}}{n\theta^2\beta} & \mbox{ if } \beta\ge\bar\beta_{J}.\ea\r.$$

(iii) When $Q=I-J$, $\pi_{i\ge2}\ge\pi_1$.
The transition occurs at  $$\bar\beta_{I-J}=1+\frac{1+n(n-2)(1+\tau)^2}{n-1}\,.$$
Then, the resulting optimal $Q$-volatility is $$\bar\Phi_Q(\bar z(\beta))=\l\{\ba{ll} \frac{n-2}{\beta-1}+\frac1{n(1+\tau)^2(\beta-1)} & \mbox{ if } \beta<\bar\beta_{I-J}\\&\\ \frac{n-2}{\beta}+\frac{1}{\beta(1+\tau)^2} & \mbox{ if } \beta\ge\bar\beta_{I-J}.\ea\r.$$ 

Assuming $n=7$ and $\tau=0.5$, Figure \ref{fig:star-phi-bar} shows the behavior of $\bar\Phi_Q(\bar z(\beta))$ when $Q=I$ (solid line), $Q=J$ (dashed line), and $Q=I-J$ (dotted line).
Vertical dashed lines represent values of $\bar\beta_Q$ for $Q$ equal to the three cases.
It is possible to prove that $\beta_Q\le\bar\beta_Q$ for $Q=I,J$, and $I-J$.
The example shows that, in general, it is not true that $\bar\beta_I\le\bar\beta_J\le\bar\beta_{I-J}$ for all $\tau\in(0,1)$.
\begin{figure}
\centering
\includegraphics{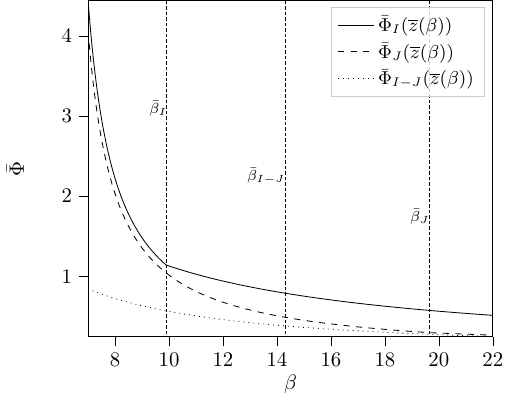}
\caption{\label{fig:star-phi-bar} Behavior of $\bar\Phi_Q(\bar z(\beta))$ for star network of Example \ref{ex:star-2} with $n=7$ and $\tau=\ups=0.5$, when $Q$ is equal to $I,J=n^{-1}\1\1'$, and $I-J$.}
\end{figure}
\end{example}

\section{Conclusion}\label{sec:conclusion}
We have studied optimal volatility reduction problems for linear network equilibrium models and idiosyncratic exogenous shocks.
We have introduced a novel general notion of target centrality and shown that the optimal network intervention has a water-filling-type structure: the nodes $i$ may be ordered by a quantity $\omega_i$ combining target centrality, shock standard deviation, and intervention cost; depending on the available budget, the system planner optimally either concentrates its efforts on the nodes with highest $\omega_i$ (when the available budget is small) or distributes its available efforts on all nodes with an intensity $z_i$ proportional to $\omega_i$ (when the available budget is sufficiently large). We then analyzed a robust version of the problem where the system planner optimizes the intervention against a worst-case vector of shocks whose variances are subject to a linear constraint. We completely characterized the solutions of this problem and showed that they also admit a water-filling-type structure.

Several directions for future work are of interest.
In a companion paper \cite{PartII}, we generalize our results to the case when the exogenous shocks hitting the network are no longer idiosyncratic but rather correlated.
Another relevant direction is to consider deterministic perturbations of network equilibrium models, as in our preliminary work \cite{damonte2022targeting}.
Yet another very interesting direction consists in moving away from network equilibrium models and considering dynamical perturbations of network systems.

\bibliographystyle{IEEEtran}
\bibliography{references2}
\begin{IEEEbiography}[{\includegraphics[width=1in,height=1.25in,clip,keepaspectratio]{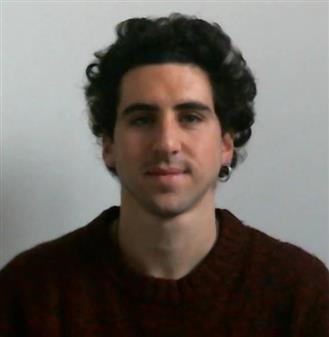}}]
{Luca Damonte} is a Post-Doc Fellow at the Department of Economics and Finance of Luiss University, Rome, Italy.
He graduated in Mathematical Engineering in 2018 from Politecnico di Torino, Italy, and obtained a Ph.D.~in Pure and Applied Mathematics in January 2023 at the same university. 
His research interests lie at the interface between network systems, optimization, and game theory, with a primary goal on investigating targeted intervention problems on socio-economic networks through optimization methods.
\end{IEEEbiography}

\begin{IEEEbiography}[{\includegraphics[width=1in,height=1.25in,clip,keepaspectratio]{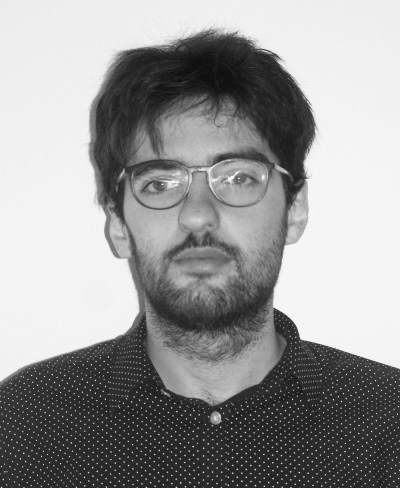}}]{Leonardo Cianfanelli} received the B.Sc.  in Physics and Astrophysics in 2014 from Università di Firenze, Italy, the M.S. in Physics of Complex Systems in 2017 from Università di Torino, Italy, and the PhD in Pure and Applied Mathematics in 2022 from Politecnico di Torino, Italy. He is currently a Research Assistant at the Department of  Mathematical Sciences, Politecnico di Torino, Italy. He was a Visiting Student at the Laboratory for Information and Decision Systems, Massachusetts Institute of Technology, in 2018--2020. His research focuses on control of network systems, with applications to epidemics, transportation, and social and productive networks.
\end{IEEEbiography}

\begin{IEEEbiography}[{\includegraphics[width=1in,height=1.25in,clip,keepaspectratio]{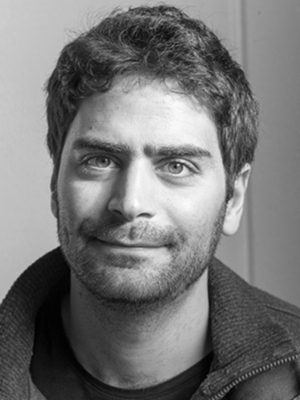}}]
{Giacomo Como}(M'12) is  a  Professor at  the Department  of  Mathematical  Sciences,  Politecnico di  Torino,  Italy. He is also a Senior Lecturer  at  the  Automatic  Control  Department, Lund  University,  Sweden.  He  received the B.Sc., M.S., and Ph.D.~degrees in Applied Mathematics  from  Politecnico  di  Torino, Italy, in  2002,  2004, and 2008, respectively. He was a Visiting Assistant in  Research  at  Yale  University  in  2006--2007  and  a Postdoctoral  Associate  at  the  Laboratory  for  Information  and  Decision  Systems,  Massachusetts  Institute of Technology in  2008--2011. Prof.~Como currently serves as Senior Editor for the \textit{IEEE Transactions on Control of Network Systems}, as Associate  Editor  for \textit{Automatica}, and as the  chair  of the  {IEEE-CSS  Technical  Committee  on  Networks  and  Communications}. He served as Associate Editor for the  \textit{IEEE Transactions on Network Science and Engineering} (2015-2021) and for the \textit{IEEE Transactions on Control of Network Systems} (2016-2022).  He served as the  IPC  chair  of  the  IFAC  Workshop  NecSys'15  and  was a  semiplenary speaker  at  the  International  Symposium  MTNS'16.  
He  is  recipient  of  the 2015  George S. ~Axelby  Outstanding Paper Award.  His  research interests  are in  dynamics,  information,  and  control  in  network  systems, with  applications to  infrastructure,  social,  and  economic networks.
\end{IEEEbiography}

\begin{IEEEbiography}[{\includegraphics[width=1in,height=1.25in,clip,keepaspectratio]{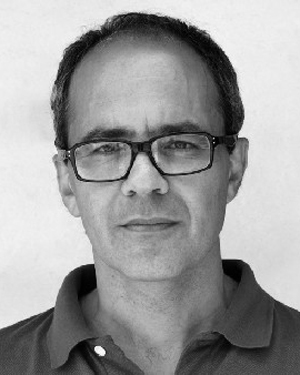}}]
{Fabio Fagnani}
received the Laurea degree in Mathematics from the University of Pisa and the Scuola Normale Superiore, Pisa, Italy, in 1986. He received the PhD degree in Mathematics from the University of Groningen,  Groningen,  The  Netherlands,  in 1991. From 1991 to 1998, he was an Assistant Professor at the Scuola Normale Superiore. In 1997, he was a Visiting Professor at the Massachusetts Institute of Technology. Since 1998, he has been with the Politecnico of Torino. From 2006 to 2012, he acted as Coordinator of the PhD program in Mathematics for Engineering Sciences at Politecnico di Torino. From 2012 to 2019, he served as the Head of the Department of Mathematical Sciences, Politecnico di Torino. His current research focuses on network systems, inferential distributed algorithms, and opinion dynamics. He is an Associate Editor of the \textit{IEEE Transactions on Automatic Control} and served in the same role for the \textit{IEEE Transactions on Network Science and Engineering} and the \textit{IEEE Transactions on Control of Network Systems}.
\end{IEEEbiography}

\end{document}